\documentclass{article}

\usepackage[utf8x]{inputenc}
\usepackage{default}
\usepackage{tikz}
\usepackage{calc}
\usepackage{amsmath,amsthm}
\usepackage{amssymb}
\usepackage{color}
\usepackage{hyperref}
\usetikzlibrary{calc}
\usetikzlibrary{intersections}
\usepackage{geometry}
\renewcommand{\d}{\text{d}}

\newcommand{\derD}[1]{\overrightarrow{#1}}
\newcommand{\derG}[1]{\overleftarrow{#1}}

\newtheorem{lemm}{Lemma}
\newtheorem{thm}[lemm]{Theorem}
\newtheorem{coro}[lemm]{Corollary}

\tikzstyle{every picture}=[level distance = 8mm, baseline=-0.5ex] 

\tikzstyle{prop}=[shape=circle,minimum size=6mm, draw=black!80, fill=green!30]

\tikzstyle{boucle}=[shape=circle,minimum size=6mm, draw=black!80]

\tikzstyle{ligne}=[shape=line, minimum size=6mm, draw=black!80]

\def\clap#1{\hbox to 0pt{\hss#1\hss}}

\newlength{\hcenteredinphantomlength}
\newlength{\hcenteredinphantomheight}

\newlength{\progressbarlength}
\newlength{\progressbarremaininglength}
\newlength{\progressbarshiftlength}
\newlength{\progressbardonethickness}
\newlength{\progressbartodothickness}
\title{Batalin--Vilkovisky formalism as a theory of integration for polyvectors}

\author{Pierre J. Clavier${}^{1}$, Viet Dang Nguyen${}^{2,3}$ \\
\\
\normalsize \it ${}^1$ Potsdam Universit\"at, Institute f\"ur Mathematik, Golm, Deutschland \\
\normalsize \it $^2$ Universit\'e Claude Bernard Lyon 1, UMR 5208, Institut Camille Jordan, 69622, Lyon, France\\
\normalsize \it $^3$ CNRS, UMR 5208, Institut Camille Jordan, 69622, Lyon, France\\
}

\date{}

\begin{document}

\maketitle

\begin{abstract}
 The Batalin--Vilkovisky (BV) formalism is a powerful generalization of the
BRST approach of gauge theories and allows to treat more general field theories.
We will see how, starting from the case of a finite dimensional
configuration space, we can see this formalism as a theory of integration
for polyvectors over the shifted cotangent bundle of the configuration
space, and arrive at a formula that admits a generalization to the infinite
dimensional case. The process of gauge fixing and the observables of the
theory will be presented.
\end{abstract}

\section{Motivations and program}

If you ask your best experimental physicist friend what the Universe is, there is a good chance that he will talk about particles; 
gluons and quarks, photons and fermions if he is studying the small structures of matter at high energies. On the other hand, if he is doing 
experimental astrophysics, he might discuss  stars, black holes and so on.

These objects are described by very different (and, to date, incompatible) theories: quantum field theory for the former, general relativity 
for the latter. However, quantum field theory and general relativity have a point in common: they happen to be gauge theories. The goal of 
the formalism devised by Batalin and Vilkovisky in \cite{BaVi77} and \cite{BaVi81b} is to deal with such theories and some of their generalizations. 
For the sake of completeness let us briefly, in non-technical terms, recall what a gauge theory is.

Let us assume that you have a $d$-dimensional space-time $\mathcal{M}^d$. A field living on that space is a function from this space-time to 
a target space which depends on the theory under consideration: it is $\mathbb{R}$ for a scalar field theory, a vector space for a vector 
field theory, and so on. The guiding principle of a gauge theory is to reparametrize your field by a ``rotation''
\begin{equation*}
 \Psi(x)\longrightarrow e^{i\theta}\Psi(x).
\end{equation*}
This is rather similar to the invariance under rotation of the wavefunctions solutions of the Schr\"odinger equation in 
usual quantum mechanics. This reparametrization is, however, a generalisation of quantum mechanics in two important ways.
\begin{itemize}
 \item The reparametrization can be made into a group more general than the group $U(1)$ of rotations. We will typically say that the reparametrization
 parameter $\theta$ is an element of some semi-simple Lie algebra $\mathfrak{g}$. This algebra will be called the gauge algebra.
 \item This $\theta\in\mathfrak{g}$ has a value that depends on the point of space-time at which we evaluate the field $\Psi$: 
 $\theta=\theta(x)$.
\end{itemize}
We call gauge theory a theory that has the invariance under such a reparametrization, plus the usual properties of any nice physical theory: 
Lorentz invariance, locality, renormalizability (in the case of quantum field theories).

Now, one of the arguably most elegant formulations of physics is known as the path-integral formalism. It was devised by Feynman in \cite{Fe48}. 
It rests upon the observation that the basic principles of quantum mechanics forbid us to determine by what slit a photon goes in the 
double slots experiment, or more precisely, it states that this is a meaningless question. Then, one can increase the number of screens with two slots 
on them between the photon's emission and detection points. For each of the screens  
there is no meaning to ask to which slot the photon went through. Therefore, in the limit of an infinite number of screens, we conclude 
that we cannot tell which \emph{path} a photon follows from $a$ to $b$: one has to make a (weighted) 
average over all possible paths, that is to perform an integration over the space of paths.

Two obvious difficulties arise. The first one is that the space of paths is huge, typically having an uncountable 
number of dimensions. Integrals are in general ill-defined in such spaces. The second one comes from the gauge freedom: if two paths 
can be mapped to each other by a gauge transformation, they correspond to the same physical path, and we shall not count both. 

Let us briefly explain the program behind the BV formalism. First, we take a finite dimensional configuration space $M$, $N=\dim(M)$. Then, 
naively, we can imagine our observables as elements of $\mathcal{C}^{\infty}(M)$. Path integral has to be an evaluation map
\begin{equation*}
 <>:\mathcal{C}^{\infty}(M)\longrightarrow\mathbb{R}.
\end{equation*}
The most natural way to do this is to choose a volume form $\Omega\in\Lambda^N T^*M$ and to define, for an observable $f\in\mathcal{C}^{\infty}(M)$
\begin{equation*}
 <f> = \int_Mf\Omega.
\end{equation*}
This has two assets. First, the integration of forms is a powerful, very well understood tool. Second, it provides an equivalence 
relation between elements of $\mathcal{C}^{\infty}(M)$ which have the same evaluation: if $f\Omega$ and $f'\Omega$ differ by a $d$-exact term (where $d$ 
is the de Rham differential), then they 
give rise to the same measured values. Hence we understand that the true observables are rather elements of the de Rham cohomology group. 
This remark gives some hope about our ability to treat gauge theories.

However, it also suffers from two serious drawbacks: the notion of a top form does not have a clear meaning in the infinite dimensional case, 
and even less so if the dimension is uncountably infinite, as it will be in the cases of interest. Moreover, $\Omega$ might not exist.

The BV formalism offers an escape road to these drawbacks. The idea is to work with polyvectors rather than with forms, since a $N-1$ 
form is the contraction of a $1$-polyvector with a $N$ form. We gain in that we can work with polyvectors even in the infinite dimensional case! 
The usual concepts of integration can be lifted to the level of polyvectors, and arrive at a formulation that admits a natural 
generalization to the infinite dimensional case. This project is summarized in figure \ref{fig1}.
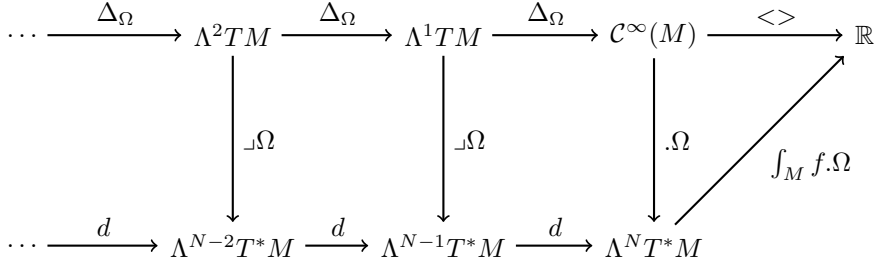
\begin{figure} \label{fig1} 
 \begin{tikzpicture}[->,shorten >=1pt,auto,node distance=2.8cm,thick] 

    \node (1) {$\mathbb{R}$};
    \node (2) [left of=1] {$\mathcal{C}^{\infty}(M)$};
    \node (3) [below  of=2] {$\Lambda^N T^*M$};
    \node (4) [left of=2] {$\Lambda^1 TM$};
    \node (5) [left of=4] {$\Lambda^2 TM$};
    \node (6) [left of=5] {$\ldots$};
    \node (7) [below of=4] {$\Lambda^{N-1} T^*M$};
    \node (8) [below of=5] {$\Lambda^{N-2} T^*M$};
    \node (9) [below of=6] {$\ldots$};

    \path
      (6) edge node [above] {$\Delta_{\Omega}$} (5)
      (5) edge node [above] {$\Delta_{\Omega}$} (4)
      (4) edge node [above] {$\Delta_{\Omega}$} (2)
      (2) edge node [above] {$<>$} (1)
      (2) edge node [right] {$.\Omega$} (3)
      (3) edge node [below right] {$\int_Mf.\Omega$} (1);

    \path
      (9) edge node [above] {$d$} (8)
      (8) edge node [above] {$d$} (7)
      (7) edge node [above] {$d$} (3);

    \path 
      (4) edge node [right] {$\lrcorner\Omega$} (7)
      (5) edge node [right] {$\lrcorner\Omega$} (8);

 \end{tikzpicture}
 \caption{a summary of the BV program.}
\end{figure}
The plan of this paper is the following: in the subsequent second section, we will define the BV integral and state some of its 
properties\footnote{in most of the case, only hints of proofs will be given, first to keep this text within a reasonable 
size, and second to clarify the idea of the construction. In practice, this will make the next section quite short}. The third 
section will make explicit how to deal with the gauge freedom within the BV formalism, and will introduce the 
Schouten--Nijenhuis brackets when working in simple (but still interesting) gauge theories. The fourth section will present the 
famous master equations, and give a few consistency checks on the set of observables built within the BV framework. The conclusion contains
a discussion on the reasons why one might be interested in studying the BV formalism.

\section{BV integral}

From now on we will use the so-called shift isomorphism.
Let $M$ be a finite dimensional vector space. The space of polyvectors is isomorphic (as a vector 
space) to the space of functions on the shifted cotangent space:
 \begin{equation}
  \Lambda TM \simeq \mathcal{C}^{\infty}(\Pi T^*M).
 \end{equation}
 We denote by $\Pi M$ the vector space $M$ shifted by $1$. Here shifted means that we reverse the parity of coordinates: fermionic coordinates are now 
 bosonic and vice-versa. Let $(e^a)_a$ be a basis of $M$. Then the $e^a$ are linear coordinates on $M^*$ and the isomorphism is given in term 
 of the basis by 
\begin{eqnarray*}
 e^{a_1}\wedge\dots \wedge e^{a_k} \mapsto \left(e^{a_1}\dots e^{e_k}\right), 
\end{eqnarray*}
(extended by linearity). The left-hand-side of the above map is a polyvector, while the right-hand-side is a polynomial. The basic idea 
behind the proof is that both sides have the same parity (that is why we have $\Pi T^*M$) and that all functions on a shifted space (which we 
call superfunctions) are polynomials. Hence, in the following, we will freely use this isomorphism, and in particular we will often write 
polyvectors as superfunctions.

 \subsection{BV laplacian}
 
In Figure \ref{fig1} we were imprecise when writing a polyvector field as an element of $\Lambda^pTM$, while 
a polyvector is instead a section of the bundle $\Lambda^pTM$ above $M$, exactly like a form is a section of the cotangent bundle 
$\Lambda^qT^*M$. This distinction is often neglected in Physics, for going from one to the other boils down to specifying a basis on 
$M$. However, we have $\alpha\in \Gamma(M,\Lambda^p(TM))$. Then the contraction operator $\lrcorner$ is a map 
\begin{equation}
 \lrcorner: \Gamma(M,\Lambda^p(TM))\times\Gamma(M,\Lambda^q(T^*M))\to\Gamma(M,\Lambda^{q-p}(T^*M))
\end{equation}
defined for $p=1$ by $(X\lrcorner\omega)(X_1,\dots,X_{q-1}):=(\iota_X\omega)(X_1,\dots,X_{q-1})=\omega(X,X_1,\dots,X_{q-1})$ and is then extended by 
$\left(\alpha\wedge\beta\right)\lrcorner\omega=\alpha\lrcorner(\beta\lrcorner\omega)$.

Now, let $\alpha$ be a polyvector\footnote{a polyvector field, to be precise but in the following we will not bother writing field everywhere} 
and $\Omega$ a well-behaved (i.e., nowhere vanishing) volume form. Then $\Delta_{\Omega}$ is the operator from 
$\Gamma(M,\Lambda^p(TM))$ to $\Gamma(M,\Lambda^{p-1}(TM)$ defined by
\begin{equation} \label{def_delta}
 (\Delta_{\Omega}\alpha)\lrcorner\Omega=d(\alpha\lrcorner\Omega).
\end{equation}
If we denote by $\mathcal{F}_{\Omega}$ the isomorphism
\begin{equation}
 \mathcal{F}_{\Omega}:\alpha \in \Gamma(M,\Lambda TM) \longmapsto \alpha\lrcorner \Omega \in \Gamma(M,\Lambda T^*M) 
\end{equation}
then $\mathcal{F}_{\Omega}^{-1}$ is well-defined since $\Omega$ was assumed to vanish nowhere. Thus
\begin{equation} \label{BV_lapl_def}
 \Delta_{\Omega} = \mathcal{F}_{\Omega}^{-1}\circ d\circ \mathcal{F}_{\Omega}
\end{equation}
where $d$ is the usual de Rham differential. The homology of this operator is described in the subsequent lemma.
\begin{lemm}
Let $\Omega$ be a volume form and $\Delta_{\Omega}$ be its associated BV laplacian. Then the operator $\Delta_{\Omega}$ defined above 
satisfies $\Delta_{\Omega}^2=0$.
\end{lemm}
In the remaining parts of this text, we will use the usual terminology and call a polyvector $\alpha$ to be $\Delta_{\Omega}$-closed when 
$\Delta_{\Omega}\alpha=0$ and $\Delta_{\Omega}$-exact when there exists a polyvector $\beta$ such that $\alpha=\Delta_{\Omega}\beta$.

While $\Delta_{\Omega}$ depends on the chosen volume form, we will show later that the observables, built as elements of the 
homology complex of $\Delta_{\Omega}$ do not. Moreover one can relate the various BV laplacians to each others, as we will see with \eqref{relation_delta}.

 \subsection{Definition of the integral}
 
Given a well-behaved volume form $\Omega$ and its associated BV laplacian $\Delta_{\Omega}$ we define the BV integral of a polyvector 
field $\alpha$ as
\begin{equation*}
 \int_{\Pi N^*\Sigma}^{BV} \alpha = \int_\Sigma \alpha\lrcorner \Omega
\end{equation*}
where $N^*\Sigma$ is the conormal of $\Sigma$, which means the part of the cotangent space $T^*M$ which vanishes on the tangent of $\Sigma$. 
It is a space of dimension $N=\dim(M)$. Indeed, if $\dim(\Sigma)=p$, then $N^*\Sigma$ has $N-p$ dimensions in the fiber.
This definition can be justified in the finite dimensional case by working out its right-hand-side with the shift isomorphism until we 
reach the left-hand-side. To make the link with the language of symplectic geometry, let us notice that conormal spaces are examples of 
Lagrangian submanifolds of $T^*M$ endowed with its natural symplectic form $dx^i\wedge dx^*_i$.

With this definition of the BV integral at hand, most of the important theorems of integration of forms can be translated as results on the integration 
of polyvectors. The first is the following counterpart to the Stokes theorem
\begin{thm} \label{thmStokes}
Let $\Omega$ be a volume form and $\Delta_{\Omega}$ be its associated BV laplacian. Let $\Sigma$ be a smooth submanifold with smooth 
boundary $\partial \Sigma$. Then for any polyvector $\alpha\in C^\infty(\Pi T^*M)$:
$$\int^{BV}_{\Pi N^*\partial\Sigma}\alpha=\int^{BV}_{\Pi N^*\Sigma}\Delta_{\Omega}\alpha.$$
\end{thm}
\begin{proof}
 \begin{align*}
  \int^{BV}_{\Pi N^*\partial\Sigma}\alpha & = \int_{\partial\Sigma}\alpha\lrcorner\Omega \quad \text{by definition of the BV integral,} \\
					  & = \int_{\Sigma}d(\alpha\lrcorner\Omega) \quad \text{by the usual Stokes theorem, and with $d$ the de Rham differential,} \\
					  & = \int_{\Sigma}(\Delta_{\Omega}\alpha)\lrcorner\Omega \quad \text{by definition of $\Delta_{\Omega}$,} \\
					  & = \int^{BV}_{\Pi N^*\Sigma}\Delta_{\Omega}\alpha.\quad \text{by definition of the BV integral.}
 \end{align*}

\end{proof}
This result has an obvious corollary 
\begin{coro} \label{gauge_fix}
 Let $\Sigma_1$, $\Sigma_2$ be two smooth submanifolds belonging to the same homology class. Then for any $\Delta_{\Omega}$-closed 
 polyvector $\alpha\in C^\infty(\Pi T^*M)$:
\begin{equation*}
 \int^{BV}_{\Pi N^*\Sigma_1}\alpha = \int^{BV}_{\Pi N^*\Sigma_2}\alpha.
\end{equation*}
\end{coro}
Finally, there is also a very simple but important lemma coming from the previous definitions.
\begin{lemm} \label{lemm_def_obs}
 If a polyvector $\alpha$ is $\Delta_{\Omega}$-exact, then 
\begin{equation*}
 \int_{\Pi N^*\Sigma}^{BV} \alpha = 0
\end{equation*}
\end{lemm}
This lemma is the equivalent to the Poincaré lemma for forms.

As for the Theorem \ref{thmStokes}, the proofs of this corollary and of this lemma are fairly easy and follow simply from playing with the definition of the 
BV integral and the usual integration of forms.

\subsection{Advantages of the BV formalism}

Having defined the BV integral, we may  start answering a simple question: what are the advantages of the BV formalism over other
approaches such as the Faddeev--Popov determinant or the BRST formalism. We will give the historical answer here. Other arguments 
in favor of the BV formalism will be given in the conclusion.

A symmetry is said to be open when it is fulfilled only on-shell, that is on the critical domain of the action $S_0$, i.e. on the  
submanifold of the configuration space where the fields are solutions to the usual equations of motion. The archetypal example 
of a physical theory with open symmetries is supergravity without auxiliary fields. As first noticed in 
\cite{Ka78}, when working in a theory with open symmetries we might end up with quartic ghost terms in the gauge-fixed lagrangian. 

In the Faddeev--Popov formalism, ghosts are interpreted as fermionic variables coming from the restriction of the domain of 
integration. This restriction is performed with delta functions, and brings a determinant, written as an integral over fermionic
variables: the ghosts. Therefore we do not have many freedom on the ghost terms that can be treated in the Faddeev--Popov formalism. 
In particular, quartic terms are not allowed, thus the Faddeev--Popov formalism is not adapted to the treatment of theories with 
open symmetries.

On the other hand, once again in the case of a theory with open symmetries, one can 
show that the BRST differential does not square to zero off-shell, hence making the BRST cohomology ill-defined. This is clearly 
explained in \cite{Weinberg96b}, where more references on the subject can be found. 
However, in sharp 
constrast with the Faddeev--Popov and BRST formalisms, since the BV formalism is a theory of integration, one can treat much more 
general functions than in the Faddeev--Popov formalism, and the well-definiteness of the integral does not depend on the precise form 
of these functions. Therefore, this formalism is particularly adapted to the treatment of theories with open symmetries.

Another reason to prefer the BV formalism over other approaches is that some questions are easier to answer with it. In particular, Jean Zinn-Justin and Laurent le 
Guillou first developed an equivalent formalism (which corresponds to the classical version of the formalism presented here) to study 
anomalies in gauge theories. Nowadays, anomalies are still often studied within the framework of the BV formalism. 

Other arguments in favor of the BV formalism will be presented in the conclusion, for they involve contemporary concerns that 
can be treated with it.

\section{Gauge fixing}

 \subsection{Gauge fixing in BV formalism}

We are now ready to handle gauge freedom. We have two types of freedom: the choice of the volume form $\Omega$, and the choice of the 
surface $\Sigma$ to integrate over. Which one corresponds to the gauge freedom, and which one is a spurious choice, in which case we will make 
sure that have the evaluation of observables does not depend on it?

To answer this question, it is useful to realize that, in the finite dimensional case, we have
\begin{equation*}
 \int_{M}e^{iS/\hbar}\Omega = \int_{\underline{0}}^{BV}e^{iS/\hbar}.
\end{equation*}
Here $\underline{0}$ is the conormal of $M$, which has its $N$ dimensions on the base space $M$, and none on the fiber i.e. it 
is the zero section of the shifted cotangent bundle $\Pi T^*M$.

But $\underline{0}$ is a very bad choice of integration domain in $\Pi T^*M$! Indeed, in the infinite dimension limit, this integral will 
diverge, due to the gauge freedom. The reason for this is that $\underline{0}$ is identified (as a Lagrangian submanifold of 
$\Pi T^*M$) with $M$, and that the gauge group $\mathfrak{g}$ is a symmetry of the theory, the integrand is left 
unchanged under the action of the group $\mathfrak{g}$. Hence, in the infinite dimensional limit, we find an integral over a non-compact 
domain of a constant.

Fortunately, Corollary \ref{gauge_fix} tells us that if an integrand is $\Delta_{\Omega}$-closed, then one can change the domain of integration 
in $\Pi T^*M$ without changing the value of the integral. More precisely, if $\Sigma$ is a smooth submanifold of $M$ and if the 
integration domain in $\Pi T^*M$ is $\Pi N^*\Sigma$ (the conormal space of $\Sigma$), then without changing the value of the integral, we can 
integrate over $\Pi N^*\Sigma'$, provided that $\Sigma'$ is in the same homology class of $M$ than $\Sigma$. Hence, the choice of a gauge 
will be the choice of a surface $\Sigma$ to integrate over in $\Pi T^*M$. And, of course, saying that a quantity $A$ is gauge invariant amounts 
to requiring that such a change of the integration surface will not change the value of the integral. Hence, the 
$\Delta_{\Omega}$-closedness
condition is the gauge-invariance condition.

Now, a question a reader may raise is: to which extent can we explicitly describe the allowed integration surfaces? Homology is a notoriously  
difficult subject, so is this definition of any practical interest? As a partial answer, we will show that we can describe a non-trivial set of 
submanifolds on which it is legitimate to integrate. We read this observation in \cite{Fi06}, although the idea was already present in the 
seminal work of Batalin and Vilkovisky \cite{BaVi81b}. Take $\Psi_1$, $\Psi_2$ two smooth functions on $M$. 
Consider the submanifolds 
$\mathcal{L}_{\Psi_j}$ of $\Pi T^*M$ defined by
\begin{equation}
 (x^i,x^*_i)\in\mathcal{L}_{\Psi_j}\Longleftrightarrow \Psi_j(x^i)=0,x_i^* = \frac{\partial\Psi_j}{\partial x^i}
\end{equation}
for $j\in\{1,2\}$. Since we can build the homotopy
\begin{equation*}
 \Psi_t = t\Psi_1 +(1-t)\Psi_2,
\end{equation*}
we find that $\mathcal{L}_{\Psi_1}$ and $\mathcal{L}_{\Psi_2}$ are in the same homology class (if $M$ is connected). Finally, since the 
zero section is defined by $x_i^*=0$ we see that, for a smooth function $\Psi$, the submanifold $\mathcal{L}_{\Psi}$ defined in the same way 
as $\mathcal{L}_{\Psi_i}$ is an admissible submanifold to be 
integrated over: it has the same homology class as the zero section.

With this picture in mind, we see that Lemma \ref{lemm_def_obs} allows to say that two gauge invariant quantities that differ only by 
a $\Delta_{\Omega}$-exact term will give the same observed values, and thus actually lead to the same observables. In more rigorous terms: the 
observables will be elements of the homology group of $\Delta_{\Omega}$.

To summarize: the choice of the surface to integrate over is the choice of a gauge in the usual formulations of gauge theories, and the 
condition of gauge invariance of a quantity $A$ is translated into $\Delta_{\Omega}(A)=0$. Hence, we will have to check that the set of 
observables that we get at the end of this procedure does not depend on the choice of the volume form $\Omega$. This will be carried out at 
the very end of this presentation.

 \subsection{Schouten--Nijenhuis bracket}
 
 Before deriving the quantum master equations, let us present the so-called Schouten--Nijenhuis bracket for a special case on which we will 
 focus from now on. We will take our configuration space $M$ to be
 \begin{equation}
  M = X\times\Pi\mathfrak{g}
 \end{equation}
 where $X$ is a space of fields (which is assumed to be bosonic), and $\mathfrak{g}$ the ghost of the theory. We will write 
 $(x^i)=(\phi^i,c^{\alpha})$ for a coordinate basis of $M$. Then the associated basis of $\Pi T^*M$ is
 \begin{equation}
  (\phi^*_i,c^*_{\alpha},\phi^i,c^{\alpha}).
 \end{equation}
There is a natural grading on this space, defined by
\begin{equation}
 |c^*_{\alpha}|=-2 \quad |\phi^*_i|=-1 \quad |\phi^i|=0 \quad |c^{\alpha}|=+1
\end{equation}
and $|f.g|=|f|+|g|$. Hence antighosts $c^*_{\alpha}$ are bosonic, while the antifields $\phi^*_i$ and the ghosts are fermionic.

Let us notice that if we had ghosts of ghosts\footnote{such objects are needed when the Koszul homology is not a resolution, i.e. has 
non-trivial homology groups other than the zeroth one. This will make the zeroth BRST cohomology not isomorphic to the set of observables 
of the theory (see \cite{Fi06} for a clear proof of this statement), and has to be cured with the Tate procedure \cite{Ta56} which will 
produce ghosts of ghosts and so on. A presentation of BRST formalism coherent with the current notations can be found in \cite{Cl15}.} they 
would be of degree $+2$, ghosts of ghosts of ghosts would be of degree $+3$, and so on...

The Schouten--Nijenhuis bracket of two polyvectors $F$ and $G$, seen as superfunctions over $\Pi T^*M$ is defined as
\begin{equation}
 \{F,G\} = \frac{\derG\partial F}{\partial x^i}\frac{\derD\partial G}{\partial x^*_i} - \frac{\derG\partial F}{\partial x^*_i}\frac{\derD\partial G}{\partial x^i}
\end{equation}
where the right derivative $\derD\partial$ is the usual derivative and the left derivative $\derG\partial$ is defined by
\begin{equation}
 \frac{\derG\partial F}{\partial y} = (-1)^{|y|(|F|+1)}\frac{\partial F}{\partial y}.
\end{equation}
This bracket has many good properties (e.g. $\{F,.\}$ is a graded derivation, it obeys a graded Jacobi identity), but they are tedious to 
show (essentially because of the powers of $-1$) and left as an exercise. However, one of their important features is their link
to the BV laplacian when written in coordinates.

\section{Master equations}

 \subsection{BV laplacians in coordinate}
 
In order to derive the link discussed above, we take the configuration space of fields $X$ to be finite dimensional, with $\dim(X)=N$. Let $(\phi_i)_{i=1,...,N}$ be a basis of 
$X$ and $\Omega_0= \d \phi^1\wedge\ldots\wedge\d \phi^N$ the associated canonical volume form.
The BV laplacian associated to this volume form (which we will write $\Delta$ from now on) reads
\begin{equation}
 \Delta = \frac{\partial}{\partial \phi^i}\frac{\partial}{\partial \phi^*_i}.
\end{equation}
The proof of this result follows directly from the definition of the BV laplacian (given by \ref{def_delta}). Then, one can compute the 
right-hand-side, and the left-hand-side of \ref{def_delta} with the proposed $\Delta$. It is very easy to see that both sides lie in the 
same directions on $\Pi T^*X$, but the hard part is to check that they are actually equal.

We can then play the same game in the space $\Pi T^*\mathfrak{g}$, leading to the same result. In order to generalize this to
$M = X\times\Pi\mathfrak{g}$, we have to merge our two partial laplacians. It turns out that the right way of doing so is to define the 
BV laplacian $\Delta$ on $M$ to be
\begin{equation} \label{BV_lapl_coord}
 \Delta = \frac{\partial}{\partial\phi^i}\frac{\partial}{\partial\phi^*_i} - \frac{\partial}{\partial c^{\alpha}}\frac{\partial}{\partial c^*_{\alpha}}.
\end{equation}
Before going further, let us notice that we defined objects that admit a natural generalization to the infinite dimensional case! 
Indeed, one would just have to replace the derivatives in the Schouten--Nijenhuis bracket and in this BV laplacian by functional derivatives 
to have operators on the configuration space of a quantum field theory. We will nevertheless go on working in the finite dimensional setup to derive the 
quantum master equations, but we have to keep in mind that what we are doing has a meaning in the infinite dimensional case.

The laplacian \eqref{BV_lapl_coord} has many interesting properties. An essential property is that the Schouten--Nijenhuis bracket measures the 
obstruction preventing the BV laplacian $\Delta$ from being a (graded) derivation. More precisely
\begin{equation} \label{BV_prod}
 \Delta(fg) = (\Delta f)g + (-1)^{|f|}f(\Delta g) + (-1)^{|f|}\{f,g\}.
\end{equation}
This property is important for it can be used as the starting point of another approach of the BV formalism, developed in particular by 
Owen Gwilliam and Kevin Costello. The reader may want to look at the thesis \cite{Gw12} and the references therein for a detailed presentation of this 
approach.
\begin{lemm} (Gwilliam \cite{Gw12})
 $\Delta$ is the unique translation invariant second order differential operator that decreases the degree of the polyvectors it acts on by 
 exactly $1$ (for the grading defined above) and satisfies \eqref{BV_prod} for $\Pi T^*X$. 
\end{lemm}
The proof of this lemma relies on the study of the constraints set by \eqref{BV_prod} on the free parameters of the most general 
translation invariant second order differential operator that decreases the degree of the polyvectors it acts on of  exactly $1$, which is
\begin{equation*}
 \Delta = a^i_{~j}(x)\frac{\partial}{\partial x^i}\frac{\partial}{\partial x^*_j} + b_i(x)\frac{\partial}{\partial x^*_i} + c_{ij}^{~~k}(x)x^*_k\frac{\partial}{\partial x^*_i}\frac{\partial}{\partial x^*_j}.
\end{equation*}
In this alternative approach, everything derives from first principles, which might make it more elegant than our down-to-earth approach. However, 
we believe that the approach presented here is more constructive in the sense that every definition is justified from physical considerations 
relative to the problem at hand.

We are now almost ready to derive the set of observables of our theory. We only need one more detail, which will be important when 
having to show that this set of observables is independent of the chosen volume form $\Omega$. Any well-behaved volume form $\Omega$ 
is still linked to the canonical volume form $\Omega_0$ by a conformal factor:
\begin{equation*}
 \Omega = e^f\Omega_0.
\end{equation*}
From the Definition \eqref{def_delta} we see that the quantity that we have to compute to get $\Delta_{\Omega}$ is
\begin{equation*}
 d(\alpha\lrcorner\Omega) = d(e^{f}\alpha\lrcorner\Omega_0)=e^{f}\left(d(\alpha\lrcorner\Omega_0) + df\wedge(\alpha\lrcorner\Omega_0)\right).
\end{equation*}
Since $e^{f}d(\alpha\lrcorner\Omega_0)=(\Delta\alpha)\lrcorner\Omega$ while 
$e^fdf\wedge(\alpha\lrcorner\Omega_0)=\{f,\alpha\}\lrcorner\Omega$, we have
\begin{equation} \label{relation_delta}
 \Delta_{\Omega} = \Delta + \{f,-\}.
\end{equation}

 \subsection{Quantum master equations}

 In order to derive the quantum master equations, let us consider the canonical volume form on the finite-dimensional space $M$. 
 Recall that a quantity is gauge invariant if, and only 
if, it is $\Delta$-closed. The very first thing to do is to check when the zero-point function 
\begin{equation*}
 \int_{M}e^{iS/\hbar}\Omega = \int_{\underline{0}}^{BV}e^{iS/\hbar}.
\end{equation*}
is gauge-invariant. This is therefore equivalent to solving the equation
\begin{equation}
 \Delta \left(e^{i\frac{S}{\hbar}}\right) = 0
\end{equation}
for $S$. This can be simplified by expanding the exponential, using the fact that $\Delta$ is a derivation and showing by induction
that $\Delta S^n = nS^{n-1}\Delta S + \frac{n(n-1)}{2}S^{n-2}\{S,S\}$. After some work we reach a very nice equation, which is 
generally called \emph{the} quantum master equation:
\begin{equation} \label{QME}
 \{S,S\} - i\hbar\Delta S=0.
\end{equation}
The classical master equation obtained from the quantum master equation by taking the limit $\hbar\rightarrow0$ reads $\{S,S\}=0$. 
This is actually an important equation, but we will not discuss it further.

We have already stated that any solution of this equation defines an action such that the zero-point function of the 
theory arising from this action is gauge invariant. Then, one can study the observables of the theory, and so on... In other words: the 
solution of this equation gives all the possible theories with a given gauge symmetry. This is why the study of the solutions of this 
Quantum Master Equation is still nowadays an active domain of research, see e.g. \cite{IgItSo07}.

Before deriving the observables of a theory, an essential consistency check of this construction is to make sure that the usual BRST action (when 
written in the language of the BV formalism) 
defines a gauge-invariant zero point function. This will be stated in a forthcoming theorem, but 
for which we need a few definitions.

Let $S_0$ be an action functional over $X$ invariant under a gauge group $\mathfrak{g}$ (let us recall that this means that 
$\mathfrak{g}$ is represented on the vector fields of $X$), that is, for any generator $e_{\alpha}$ 
\begin{equation*}
 \rho(e_\alpha)S_0=0,
\end{equation*}
where $\rho:\mathfrak{g}\mapsto \Gamma(X,TX)$ is the representation function of $\mathfrak{g}$. Then the BV action is
\begin{equation}
 S=S_0+S_E+S_R=S_0+\underset{S_E}{\underbrace{c^\alpha c^\beta C_{\alpha\beta}^\gamma c^*_\gamma}} + \underset{S_R}{\underbrace{\rho_\alpha^i c^\alpha x_i^*}}.
\end{equation}
One can view $S$ as a perturbation of $S_0$ in $\Pi T^*M$\footnote{someone working in homology theory may ponder about this terminology}. 

To state the link between this action and the BRST formalism, let us say that $S_E+S_R$ is the symbol of the BRST differential (seen as a 
vector field) in $\Pi T^*M$. The symbol of a vector field $Q$ on a smooth manifold $V$ of cotangent bundle $\pi:T^*V\mapsto V$
is defined as a function $\sigma(Q)\in T^*V$ such that the Hamiltonian vector field $\{\sigma(Q),.\}$ is a vector field on $T^*V$ whose 
horizontal part is just $Q$. In other words $\forall f\in C^\infty(V),\{\sigma(Q),\pi^*f\}=Qf$.

\begin{thm} \label{to_be_proved}
Let $S=S_0+S_E+S_R$ the action defined above and $\Delta$ the BV Laplacian corresponding to the canonical volume form $\Omega_0$, if
\begin{itemize}
\item[$\bullet$] the Lie algebra $\mathfrak{g}$ acts on $X$ in such a way that it preserves the measure $\Omega_0$
\item[$\bullet$] the Lie algebra $\mathfrak{g}$ is unimodular,
\end{itemize}
then $S$ is a solution of the Quantum Master Equation.
\end{thm}
The proof of this theorem is quite long and far beyond the scope of this introduction. It can be found in many places in the literature, 
and in particular in \cite{Cl15}. However, interestingly, we have $\Delta S=\{S,S\}=0$. Hence $S$ is a solution of the 
classical master equation as well as the quantum master equation: we say that $S$ is non-anomalous.

Indeed, the classical and quantum master equations encode, in a very deep sense, the symmetries of the theory. Hence, the classical and the 
quantum version of the theory have the same symmetries, which is exactly the definition of what non-anomalous mean. For completeness, let 
us point out that in the BV formalism, there is a weaker version of non-anomalous, which is just that we can find a completion (order by order 
in $\hbar$) of a solution of the classical master equation such that this completion fulfills the quantum master equation.

\subsection{Observables in BV formalism}

Assuming that we have found an action $S$ solution of the Quantum Master Equation, what are the observables of the theory? First of all, 
they have to be gauge invariant quantities so, according to the previous explanations
$\int^{BV}\mathcal{F}e^{iS/\hbar}$ has to be well-defined, that is
\begin{equation*}
 \Delta\left(\mathcal{F}e^{iS/\hbar}\right) = 0.
\end{equation*}
As for the Quantum Master equation, expanding the exponential we eventually found that this amounts to $\mathcal{F}$ being 
a solution of
\begin{equation} \label{qme_obs}
 \{S,\mathcal{F}\} - i\hbar\Delta\mathcal{F} = 0.
\end{equation}
This equation is the other quantum master equation, and justifies the plural form in the title of this section. Now, observables will be 
equivalent classes of solutions of this equation. Indeed, if 
$\Delta\left(\mathcal{F}e^{iS/\hbar}\right) = \Delta\left(\mathcal{F'}e^{iS/\hbar}\right) = 0$ and if 
$(\mathcal{F}-\mathcal{F}')e^{iS/\hbar}=\Delta\mathcal{G}$ then for any Lagrangian submanifold $\Sigma$
\begin{equation*}
 \int^{BV}_{\Pi N^*\Sigma}(\mathcal{F}-\mathcal{F}')e^{iS/\hbar} = 0,
\end{equation*}
i.e. $\mathcal{F}$ and $\mathcal{F}'$ lead to the same measured values. Hence, we see that the observables of the theory are instead elements 
of the zeroth homology group of 
\begin{equation}
 \mathcal{O} = \{S,-\}-i\hbar\Delta.
\end{equation}
Finally, our last task is to check that these are independent of the choice of the volume form $\Omega$.

So, let us assume that we have chosen an action $S$ solution of the Quantum Master Equation. We want to know what are the observables of 
the theory, deriving them for instance from the volume form
\begin{equation*}
 \Omega = e^{iS/\hbar}\Omega_0.
\end{equation*}
With respect to this volume form we rewrite the integral to be computed
\begin{equation*}
 \int^{BV,\Delta}_{\Pi N^*\Sigma}\mathcal{F}e^{iS/\hbar} = \int^{BV,\Delta_{\Omega}}_{\Pi N^*\Sigma}\mathcal{F} := \int_{\Sigma}\mathcal{F}e^{iS/\hbar}\lrcorner\Omega_0,
\end{equation*}
where we have explicitely written the BV laplacian with which every BV integral is built. Therefore, for the right-hand-side of this 
equation, the gauge invariance condition for $\mathcal{F}$ reduces to $\Delta_{\Omega}\mathcal{F}=0$.


Before we check that this is equivalent to the quantum master equation \eqref{qme_obs}, let us notice that until now, we have made sure that
the two integrals are the same in the finite dimensional case, and we just have to check that this equality is preserved when on the level of  
the quantum master equations. This is motivated by the fact that these 
quantum master equations are the starting point (in the BV approach) for defining observables in quantum field theory.

From the relation \eqref{relation_delta} we see that
\begin{equation}
 \Delta_\Omega\mathcal{F} = 0 \Leftrightarrow \Delta\mathcal{F} + \{iS/\hbar,\mathcal{F}\} = 0 \Leftrightarrow \{S,\mathcal{F}\} - i\hbar\Delta\mathcal{F} = 0.
\end{equation}
Hence the set of observables is independent of the choice of a volume form! Now, we could have taken a more general volume form
\begin{equation*}
 \Omega = e^f\Omega_0,
\end{equation*}
in which case the zero point function would be 
\begin{equation*}
 \int^{BV}_{(\Omega)}e^{-f}e^{iS/\hbar} := \int\mathcal{F}e^{iS/\hbar}\lrcorner\Omega_0 = \int^{BV}_{(\Omega_0)}e^{iS/\hbar}.
\end{equation*}
Then one can check that
\begin{equation*}
 \Delta_{\Omega}\left(e^{-f}e^{iS/\hbar}\right) = 0 \Leftrightarrow \{S,S\} - i\hbar\Delta S = 0.
\end{equation*}
One moreover checks that the set of observables is independent of $f$, that is
\begin{equation*}
 \Delta_{\Omega}\left(e^{-f}e^{iS/\hbar}\mathcal{F}\right) = 0 \Leftrightarrow \{S,\mathcal{F}\} - i\hbar\Delta\mathcal{F} = 0
\end{equation*}
if $S$ is a solution of the Quantum Master Equation. This is only slightly more cumbersome than the case we treated in more details above.

\section{Conclusion}

We have now presented the basics idea behind our approach of the BV formalism. Many technical details were left behind the curtain, but we 
hope that the main philosophy was somehow conveyed to the reader. Let us briefly summarize it here, before saying a word on a few questions that 
are left unanswered in this text.

The BV formalism aims to deal with theories having a gauge symmetry. In such theories, the action is invariant under the action of a group (the 
gauge group). Thus the integrand is constant over non-compact submanifolds of the configuration space, making the path integral 
ill-defined. The BRST solution to this problem is to quotient out the orbits of the gauge group and to declare that the real 
physically relevant configuration space is the quotiented space.

The BV approach is quite different. The idea is to work in some extended configuration space, that is on the shifted cotangent space of the 
initial configuration space, and to view the path integral (in the finite dimensional case) as an integral over some lagrangian 
submanifold of this extended space. One can change the integration domain without changing the value of the integral if the integrand 
has a nice property (which is just gauge invariance) and if the new integration domain is in the same homology class than the former.

The possible gauge-invariant actions and the possible gauge-invariant observables of the theory defined by 
such an action are shown to be independent of the choice of a volume form in the finite dimensional case. The remarkable point is that 
formulas derived in this way admit a natural generalization when the configuration space is infinite dimensional, which is of course the case 
of interest in quantum field theory.

These infinite dimensional generalizations can be transposed to be the definition of the observables in quantum field theory. This is the same 
philosophy as for the BRST approach of quantum field theory: in the finite dimensional case, one shows that the observables are the 
elements of a certain cohomology group. This cohomology group is still well-defined in the infinite dimensional case, and we take it to be 
the definition of the observables in quantum field theory.

Now, a reader may ask: what happens in practice when moving to the infinite dimensional case? Well, first, the coordinates $x^i$ will become 
fields, so that any derivative (for example in the BV laplacian) will become a functional derivative. Whereas in the finite dimensional case,
repeated indices meant a summation over the basis elements of the configuration space, in the infinite dimensional case, this has to be 
understood as an integral over the space-time one is dealing with as well as a discrete summation over the gauge group generators in the 
case of $c^{\alpha}c^*_{\alpha}$. Then other usual features of quantum field theories can be treated with these objects. For example, the 
issue of renormalization has been dealt with since the beginning of BV formalism, but a modern point of view on the subject can be 
found in \cite{Costello}.

We have given above some arguments in favor of the BV formalism, namely that it allows to treat more general theories than the 
Faddeev--Popov and BRST formalisms, and that it is more adapted to some questions, e.g. the study of anomalies in QFT. Moreover,
we have recently observed a revival of the BV formalism, which has found applications in various areas.

In particular, it has been used (together with its hamiltonian formulation: the Batalin--Fradkin--Vilkovisky (BFV) formalism) to 
compute the Chern--Simons invariants of manifolds with boundaries. The article \cite{CaWeMn15} of the present volume gives a 
detailled account of this line of research. Moreover, it has also recently been argued in \cite{BrFrRe13} that the BV formalism written in the 
language of category theory might allow to work out perturbative quantum gravity as a perturbative quantum field theory. However, we 
would like to finish this text stressing the intrinsic elegance of the BV theory, which is one of the main reasons that motivated 
our interest for it.

\paragraph{Acknowledgements:}

We would like to thank the organizers of this very interesting and pleasant school as well as all the participants and speakers who 
have clearly contributed to the quality of the school. Moreover, the authors would like to thank Christian Brouder and Fr\'ed\'eric 
H\'elein for many enlightning discussions.

\bibliographystyle{unsrt}
\bibliography{thesis}

\begin{thebibliography}{10}

\bibitem{BaVi77}
I.A. Batalin and G.A. Vilkovisky.
\newblock Relativistic {S}-matrix of dynamical systems with boson and fermion
  constraints.
\newblock {\em Physics Letters B}, 69(3):309 -- 312, 1977.

\bibitem{BaVi81b}
I.A. Batalin and G.A. Vilkovisky.
\newblock Gauge algebra and quantization.
\newblock {\em Physics Letters B}, 102(1):27 -- 31, 1981.

\bibitem{Fe48}
Richard~P. Feynman.
\newblock Relativistic {C}ut-{O}ff for {Q}uantum {E}lectrodynamics.
\newblock {\em Phys. Rev.}, 74:1430--1438, Nov 1948.

\bibitem{Ka78}
R.E. Kallosh.
\newblock Modified {F}eynman rules in {S}upergravity.
\newblock {\em Nuclear Physics B}, 141(1):141 -- 152, 1978.

\bibitem{Weinberg96b}
Steven Weinberg.
\newblock {\em The Quantum Theory of Fields}, volume~2.
\newblock Cambridge University Press, 1996.

\bibitem{Fi06}
Jos\'e Figueroa-O'Farrill.
\newblock {BRST} cohomology.
\newblock 2006.

\bibitem{Ta56}
John Tate.
\newblock Homology of {N}oetherian rings and local rings.
\newblock {\em Illinois J. Math.}, 1(1):14 -- 27, 1957.

\bibitem{Cl15}
Pierre~J. Clavier.
\newblock {\em Analytic and geometrical approaches of non-perturbative quantum
  field theories}.
\newblock PhD thesis, Université Pierre et Marie Curie, 2015.

\bibitem{Gw12}
Owen Gwilliam.
\newblock {\em Factorization algebra and free field theories}.
\newblock PhD thesis, 2012.

\bibitem{IgItSo07}
Yuji Igarashi, Katsumi Itoh, and Hidenori Sonoda.
\newblock Quantum {M}aster {E}quation for {QED} in {E}xact {R}enormalization
  {G}roup.
\newblock {\em Prog. Theor. Phys.}, 118:121--134, 2007.

\bibitem{Costello}
Kevin Costello.
\newblock {\em Renormalization and Effective Field Theory}.
\newblock AMS, 2011.

\bibitem{CaWeMn15}
Alberto~S. Cattaneo, Konstantin Wernli, and Pavel Mnev.
\newblock Split {C}hern-{S}imons theory in the {BV}-{BFV} formalism.
\newblock 2016.

\bibitem{BrFrRe13}
Romeo Brunetti, Klaus Fredenhagen, and Katarzyna Rejzner.
\newblock Quantum gravity from the point of view of locally covariant quantum
  field theory.
\newblock 2013.

\end{thebibliography}

\end{document}